\documentclass[11pt, english]{article}
\usepackage[margin=1in,bottom=1in,top=1in]{geometry}

\usepackage[dvipsnames]{xcolor}
\usepackage[square,sort,comma,numbers]{natbib}
\usepackage{amsthm}
\usepackage{amsmath}
\usepackage{amssymb}
\usepackage{setspace}
\usepackage{mathtools}
\usepackage{graphicx}
\graphicspath{ {./images/} }
\usepackage[hidelinks]{hyperref}
\usepackage{cleveref}
\usepackage{bbold}
\usepackage{hyperref}
\usepackage[shortlabels]{enumitem}
\usepackage{framed}
\usepackage{subcaption}
\usepackage{xspace}
\usepackage{thmtools} 
\usepackage{thm-restate}
\usepackage[many]{tcolorbox}
\usepackage{tikz-feynman}
\usepackage[linesnumbered,boxed, vlined]{algorithm2e}
	\DontPrintSemicolon
  \SetKwInOut{Input}{Input}
  \SetKwInOut{Output}{Output}
  \SetKwInOut{Assumption}{Assumption}

  \SetKwProg{function}{Function}{:}{}

\makeatletter
\newcommand*{\algotitle}[2]{%
  \stepcounter{algocf}%
  \hypertarget{algocf.title.\theHalgocf}{}%
  \NR@gettitle{#1}%
  \label{#2}%
  \addtocounter{algocf}{-1}%
}
\makeatother

\usepackage{thmtools}
\usepackage{ifthen} 
\usepackage{floatrow}

\usepackage{tikz}
\usepackage{mathdots}
\usepackage{xcolor}
\usepackage{diagbox}
\usepackage{colortbl}
\usepackage[absolute,overlay]{textpos}

\usetikzlibrary{shapes.misc}
\usetikzlibrary{decorations.pathmorphing}

\tikzset{snake it/.style={decorate, decoration=snake}}
\usetikzlibrary{math}

\usetikzlibrary{calc}
\usetikzlibrary{decorations.pathreplacing}
\usetikzlibrary{positioning,patterns}
\usetikzlibrary{arrows,shapes,positioning}
\usetikzlibrary{decorations.markings}

\tikzstyle{edge}=[very thick]
\definecolor{bostonuniversityred}{rgb}{0.8, 0.0, 0.0}
\definecolor{arsenic}{rgb}{0.23, 0.27, 0.29}
\tikzstyle{diredge}=[postaction={decorate,decoration={markings,
		mark=at position .95 with {\arrow[scale = 1]{stealth};}}}]
\tikzset{
    arrow/.style={decoration={markings, mark=at position 0.7 with
    {\fill(-0.09*#1,-0.03*#1) -- (0,0) -- (-0.09*#1,0.03*#1) -- cycle;}}, postaction={decorate}},
    arrow/.default=1
}
\tikzset{
    arow/.style={decoration={markings, mark=at position 1 with
    {\fill(-0.09*#1,-0.03*#1) -- (0,0) -- (-0.09*#1,0.03*#1) -- cycle;}}, postaction={decorate}},
    arow/.default=1
}
\tikzset{
    arrrow/.style={decoration={markings, mark=at position 0.9 with
    {\fill(-0.09*#1,-0.03*#1) -- (0,0) -- (-0.09*#1,0.03*#1) -- cycle;}}, postaction={decorate}},
    arow/.default=1
}

\newcommand{\fitellipsis}[2] 
{\draw [fill=white]let \p1=(#1), \p2=(#2), \n1={atan2(\y2-\y1,\x2-\x1)}, \n2={veclen(\y2-\y1,\x2-\x1)}
    in ($ (\p1)!0.5!(\p2) $) ellipse [ x radius=\n2/2+0cm, y radius=1.1cm, rotate=\n1];
}
\newcommand{\Fitellipsis}[2] 
{\draw [fill=white]let \p1=(#1), \p2=(#2), \n1={atan2(\y2-\y1,\x2-\x1)}, \n2={veclen(\y2-\y1,\x2-\x1)}
    in ($ (\p1)!0.5!(\p2) $) ellipse [ x radius=\n2/2+0cm, y radius=1.4cm, rotate=\n1];
}

\theoremstyle{plain}

\newtheorem*{thm*}{Theorem}
\newtheorem{thm}{Theorem}[section]
\Crefname{thm}{Theorem}{Theorems}

\newtheorem*{lem*}{Lemma}
\newtheorem{lem}[thm]{Lemma}
\Crefname{lem}{Lemma}{Lemmas}

\newtheorem*{assumption*}{Assumption}

\newtheorem{claim}{Claim}

\Crefname{prop}{Proposition}{Propositions}

\Crefname{remar}{Remark}{Remarks}

\crefname{cor}{Corollary}{Corollaries}

\newtheorem*{conj*}{Conjecture}

\crefname{conj}{Conjecture}{Conjectures}

\Crefname{qn}{Question}{Questions}

\Crefname{obs}{Observation}{Observations}

\Crefname{ex}{Example}{Examples}

\theoremstyle{definition}

\Crefname{prob}{Problem}{Problems}

\newtheorem{defn}[thm]{Definition}
\Crefname{defn}{Definition}{Definitions}

\theoremstyle{remark}

\makeatletter

\renewenvironment{proof}[1][]{\begin{trivlist}
\item[\hspace{\labelsep}{\bf\noindent Proof#1.\/}] }{\qed\end{trivlist}}

\newcommand{\remove}[1]{}

\newcommand{\Sat}{\mathrm{Sat}}
\newcommand{\Low}{\mathrm{Low}}

\newcommand{\din}{\mathrm{in}}
\newcommand{\dout}{\mathrm{out}}

\title{Edge-disjoint paths in expanders: online with removals}
\date{}

\author{
Nemanja Dragani\'c\thanks{
Department of Mathematics, ETH, Z\"urich, Switzerland. Research supported in part by SNSF grant 200021\_196965.
\emph{email}: \textbf{nemanja.draganic@math.ethz.ch}.
}
\and
Rajko Nenadov\thanks{School of Computer Science, University of Auckland, New Zealand. \emph{email}: \textbf{rajko.nenadov@auckland.ac.nz}.}
}
\begin{document} 
\maketitle
\begin{abstract}
We consider the problem of finding edge-disjoint paths between given pairs of vertices in a sufficiently strong $d$-regular expander graph $G$ with $n$ vertices. In particular, we describe a deterministic, polynomial time algorithm which maintains an initially empty collection of edge-disjoint paths $\mathcal P$ in $G$ and fulfills any series of two types of requests: 
\begin{enumerate}
    \item Given two vertices $a$ and $b$ such that each appears as an endpoint in $O(d)$ paths in $\mathcal P$ and, additionally, $|\mathcal P|  = O(n d / \log n)$, the algorithm finds a path of length at most $\log n$ connecting $a$ and $b$ which is edge-disjoint from all other paths in $\mathcal P$, and adds it to $\mathcal P$.
    \item Remove a given path $P \in \mathcal{P}$ from $\mathcal{P}$.
\end{enumerate}
Importantly, each request is processed before seeing the next one. The upper bound on the length of found paths and the constraints are the best possible up to a constant factor. This establishes the first online algorithm for finding edge-disjoint paths in expanders which also allows removals, significantly strengthening a long list of previous results on the topic.
\end{abstract}

\section{Introduction}
Finding a collection of pairwise edge-disjoint paths which connect prescribed pairs of vertices $(a_i, b_i)_{i \in [r]}$ in a graph $G$ is a classical and extensively studied problem in computer science. It is an NP-complete problem which becomes tractable when $r$ is fixed \cite{robertson95disjointpaths}. In the case of directed graphs, the problem remains NP-complete even for $r = 2$ \cite{fortune80directed}. 

In this paper we focus on the extensively studied instance of this problem when $G$ is a sufficiently strong $d$-regular \emph{expander} on $n$ vertices. Peleg and Upfal~\cite{peleg1987constructing} showed that the problem on such graphs becomes tractable for significantly larger values of $r$, under an assumption that all the given pairs $(a_i, b_i)$ are pairwise disjoint. In particular, they showed that any set of at most $r = O(n^{c})$ disjoint pairs of vertices can be connected by edge-disjoint paths, for some $c < 1/2$ which depends on the expansion properties of $G$, and such paths can be found in polynomial time. Let us briefly discuss why this is not surprising and establish a (theoretical) upper bound of $r$ that can be attained.

\paragraph{Routing in expanders.} Without going into detail of the definition of $d$-regular expanders (Definition \ref{def:expander}), let us note that one of their main features is that the diameter is $O(\log n)$ even if we remove, say, up to $d/3$ edges touching each vertex. Therefore, if we can `nicely' distribute paths (i.e.\ no vertex belongs to too many) between $(a_1, b_1), \ldots, (a_{r-1}, b_{r-1})$, then this observation immediately tells us that we can connect $a_r$ and $b_r$ using a path of length $O(\log n)$. As there exist expanders in which most pairs of vertices are actually at distance $\Theta(\log n)$, this is clearly a length one cannot avoid if the pairs are chosen adversarially. Moreover, there also exist expanders in which the second shortest path between \emph{every} two vertices is of length $\Omega(\log n)$ (i.e.\ expanders with high girth), which provides further evidence that shorter paths, even if they exist, are not feasible. Therefore, if we restrict to sets of disjoint pairs, the largest $r$ one can hope for is of order $n / \log n$. In a more general case where we do not impose such a constraint, as $G$ has $nd/2$ edges, the largest $r$ is (or could be) of order $nd / \log n$.

\paragraph{Previous results.} The result of Peleg and Upfal falls short of the theoretically best possible bound on $r$, and it was an important open problem to determine whether such a bound is (algorithmically) attainable. Starting with Broder, Frieze and Upfal \cite{broder94existence}, who improved $r$ to $\Theta(n / \log^C n)$, for some $C \ge 7$ (which, again, depends on the expansion properties), this problem has attracted significant attention. The bound on $r$ was further improved by Leighton and Rao \cite{leighton1996circuit} (see \cite[Section 3.18]{leighton99multiflow}), Broder, Frieze and Upfal \cite{broder97existense}, and Leighton, Lu, Rao, and Srinivasan \cite{leighton1998multicommodity,leighton2001new}. Finally, Frieze~\cite{frieze01final} gave a randomised polynomial time algorithm which connects any pair of $\Theta(nd / \log n)$ pairs of vertices, as long as each vertex appears as an endpoint at most $\varepsilon d$ times, for some constant $\varepsilon > 0$. In the case of directed expanders, the same result was established later by Bohman and Frieze \cite{bohman03arc}. We note that in the case where $G$ is a random graph with average degree $d$, an even better bound on $r$ of order $nd / \log_d n$ was obtained by Broder, Frieze, Suen, and Upfal \cite{broder98random} and Frieze and Zhao \cite{frieze99randomreg} (note that the diameter of a random graph can be somewhat smaller than the diameter of an expander graph, which results in slightly larger $r$). As the latest result on this topic, Alon and Capalbo~\cite{alon2007finding} have further improved the undirected case by designing a deterministic polynomial time algorithm which not only can deal with the optimal number of requested pairs, but their algorithm is also \emph{online} in the sense that the pairs are given one by one, and the algorithm has to find a path between the current pair of vertices before seeing the next one (once a path is established, it cannot be altered).

A related problem of finding vertex-disjoint paths in random graphs was considered by Shamir and Upfal \cite{shamir85disjoint}, Hochbaum \cite{Hochbaum92}, Broder, Frieze, Suen, and Upfal \cite{broder96vertexdisjoint} and, quite recently, Dragani\'c, Krivelevich, and Nenadov \cite{draganic2022rolling}.

\paragraph{Our contribution.} We generalise all these results by presenting an online deterministic algorithm with removals: The client can request a new pair of vertices to be connected, or a previously established path to be removed. The client can continue with requests indefinitely as long as the total number of (active) paths is at most $O(n d/\log n)$ in case of expanders, that is, at most $O(n d / \log_d n)$ in the case of the so-called \emph{$(n, d, \lambda)$-graphs} with $\lambda < d^{1 - \varepsilon}$ (which covers the case of random $d$-regular graphs). In addtion, we always guarantee that the length of the found paths is at most of order $\log n$, which resolves a question of Alon and Capalbo~\cite{alon2007finding}. Our approach significantly deviates from the methods used in the aforementioned work. It is conceptually simple and exploits the previously discussed, fundamental reason why finding many edge-disjoint paths is possible -- even if we only have a subset of edges at our disposal, the diameter is still $O(\log n)$.

The paper is organised as follows. In the next section we formally state our results. Section \ref{sec:oracle} describes a data structure that we call \textsc{Edge-Oracle}, which is the heart of the algorithm. Instead of working with all available edges to find a path between two given vertices, this data structure carefully chooses, in an online fashion, a subset of edges which makes sure that the found paths are well distributed. Finally, in Section \ref{sec:main_result} we describe the algorithm for establishing new paths and prove its correctness. The algorithm is largely based on Breadth-First Search, with the edges used to explore the graph obtained via \textsc{Edge-Oracle}.

\subsection{Online routing with removals}

We say that a directed graph (\emph{digraph} for short) is \emph{$d$-regular} if the in- and out-degree of every vertex is $d$. Given a $d$-regular graph $G$ (directed or undirected) and $r \in \mathbb{N}$, we define the \emph{$r$-Routing game} on $G$ played by two players, \emph{Adversary} and \emph{Router}, as follows. Throughout the game, Router maintains a family of pairwise edge-disjoint paths $\mathcal P$, which is initially empty. In each turn, Adversary has two types of requests:
\begin{itemize}
    \item \textsc{Find-Path $a$ $b$}: Given vertices $a$ and $b$ in $V(G)$ such that both $a$ and $b$ appear as endpoints of less than $d/200$ paths in $\mathcal{P}$ and, additionally, $|\mathcal P| < r$, Router is required to find a path from $a$ to $b$ (this path has to be directed from $a$ to $b$ if $G$ is a directed graph) which is edge-disjoint from all the other paths in $\mathcal P$. The found path is then added to $\mathcal{P}$. If such a path does not exist, Router loses the game.
    
    \item \textsc{Remove-Path $P$}: A path $P \in \mathcal P$ is removed from $\mathcal{P}$.
\end{itemize}
It is important to note that each request has to be fulfilled by Router before seeing the next one. We say that Router wins the $r$-Routing game if she can satisfy any (infinite) series of requests. 


\begin{defn}[Expander graphs] \label{def:expander}
We say that a $n$-vertex $d$-regular graph $G$ is an \emph{$(\beta, \gamma)$-expander}, for some $\beta, \gamma > 0$, if for every $S \subseteq V(G)$, $|S| \le n/2$, we have
$$
    e_G(S) \le \begin{cases}
        \gamma d |S| & \text{if $|S| \le \beta n$} \\
       d |S|/3  & \text{if $\beta n < |S| \le n/2$.}
    \end{cases}
$$
A $d$-regular digraph $D$ is a \emph{$(\beta, \gamma)$-expander} if the $2d$-regular (multi-)graph $D_\emptyset$ obtained from $D$ by ignoring the directions of the edges is a \emph{$(\beta, \gamma)$-expander}.
\end{defn}

The bound on the number of edges within large subsets is somewhat arbitrary and chosen for convenience, akin to the one used by Alon and Capalbo \cite{alon2007finding}. The following is our main theorem.

\begin{thm}\label{thm:main}
Let $G$ be a $d$-regular $(\beta, \gamma)$-expander graph (directed or undirected) with $n$ vertices, for $200< d < n$, positive constants $\beta$ and $\gamma<1/1000$, and $n$ large enough. Then Router has a strategy to win the $r$-Routing game for $r = \varepsilon n d / \log n$, for a constant $\varepsilon = \varepsilon(\beta, \gamma) > 0$. Moreover, Router can respond to each request in $O(n^3 d^3)$ time and each path in $\mathcal P$ is of length $O(\log n)$.
\end{thm}

Let us remark that our main goal was to: (a) show that a strategy for winning an $r$-Routing game, for $r = \Theta(nd / \log n)$, exists, and (b) that it can be implemented in polynomial time. We leave it for future work to improve the time complexity. Moreover, the chosen constants could be improved with a more meticulous analysis; we made a conscious decision not to pursue this in order to maintain clarity and simplicity in our presentation.


When applied on a graph with stronger guarantees on its edge distribution, the algorithm underpinning Theorem \ref{thm:main} allows for an improved bound on $r$ and the length of found paths. This is summarised in the following result.


\begin{defn}
We say a graph $G$ is an $(n, d, \lambda)$-graph if it is a $d$-regular graph with $n$ vertices and the second largest absolute eigenvalue of its adjacency matrix is at most $\lambda$.
\end{defn}

It is well-known that, when $\lambda$ is bounded away from $d$ , $(n, d, \lambda)$-graphs have stronger expansion properties than the ones given in Definition \ref{def:expander} (see \cite[Section 9.2]{alon2016probabilistic}). A random $d$-regular graph with $n$ vertices, for example, is with high probability an $(n, d, \Theta(\sqrt{d}))$-graph (see~\cite{friedman1991second}). Moreover, there exist explicit constructions of $(n, d, \Theta(\sqrt{d}))$-graphs, which include the famous Ramanujan graphs \cite{lubotzky1988ramanujan}. 

\begin{thm}\label{thm:ndL}
Let $G$ be an $(n, d, \lambda)$-graph for $200< d < n$, where $\lambda < \varepsilon d$ for a small enough constant $\varepsilon>0$. Then Router has a strategy to win the $r$-Routing game for $r = \alpha n d \frac{\log(d / \lambda)} {\log n}$, for some absolute constant $\alpha > 0$. Moreover, Router can respond to each request in $O(n^3 d^3)$ time and each path in $\mathcal P$ is of length $O\left(\frac{\log n}{
 \log (d / \lambda)}\right)$.
\end{thm}

Finally, let us note that for our results we do not use the fact that a (di)graph $G$ is $d$-regular in an essential way, and the same algorithm would work assuming that each degree is within $(1 \pm \varepsilon)d$. The choice to work with regular graphs is, again, purely for simplicity.

\subsection{Notation}

If $G$ is a (directed) graph, we denote with $e_G(S)$ the number of edges with both endpoints in $S$, for a given $S \subseteq V(G)$. Let $D$ be a directed graph, and let $S_1,S_2 \subseteq V(D)$ be subsets of its vertices. By $\textrm{out}_D(S_1,S_2)$ we denote the number of edges $(u,v)$ with $u\in S_1$ and $v\in S_2$, 
and by $\textrm{in}_D(S_1,S_2)$ the number of edges $(u,v)$ with $u\in S_2$ and $v\in S_1$. We write $\textrm{out}_D(S_1)$ to denote $\textrm{out}_D(S_1,V(D))$.
We denote by $\textrm{Out}_D(S_1)$ the set of vertices $v\notin S_1$ such that there exists a vertex $u\in S_1$ with $(u,v)\in D$. Similarly, $\textrm{In}_D(S_1)$ denotes the set of vertices $v \notin S_1$ such that there exists an edge $(v, u) \in D$ for some $u \in S_1$. We write $\dout_D(v)$ and $\din_D(v)$ for the out- and in- degree of $v$. A digraph is $d$-regular if it both the out- and in-degree of every vertex is $d$. By $\overleftarrow{D}$ we denote the digraph obtained from $D$ by reversing the direction of each of its edges.

\section{Edge-Oracle: a data structure for accessing edges}

\label{sec:oracle}


The following lemma shows that one can choose, in an online manner, a subset of edges in a digraph $D$ with a significant imbalance between in- and out-degrees of some vertices. Why this is useful will become apparent in the next section, where we prove Theorem \ref{thm:main}. For now, let us remark that the lemma can be viewed as a generalization of a result of Aggarwal et al.~\cite[Theorem 2.2.7]{aggarwal96efficient} which, in turn, is an algorithmic version of an earlier result by Feldman, Friedman and Pippenger \cite{feldman88nonblocking}. 

\begin{lem} \label{lemma:oracle}
Suppose $D$ is a $d$-regular digraph with $n$ vertices, for some $d \ge 10$, such that for every $S \subseteq V(D)$ of size $|S| \le \beta n$ we have
$$
e_D(S) \le \gamma |S| d,
$$
for $\gamma\leq 1/50$. Then there exists a data structure, dubbed \textsc{Edge-Oracle$(D)$}, which maintains an initially empty set of edges $H \subseteq D$ and supports the following two requests:
\begin{itemize}
    \item \textsc{add-edge $v$}: Given a vertex $v \in V(D)$ with $\dout_H(v) < \lfloor d/2 \rfloor$, the data structure returns an out-edge $e = (v,w) \in D \setminus H$ of $v$ such that $\din_H(w) < \lfloor d/5 \rfloor$, and adds it to $H$.\\
    \textbf{Note:} this request is only allowed if $|H| < \beta d n / 120$.
    
    \item \textsc{remove-edge $e$}: Remove a given edge $e \in H$ from $H$.
\end{itemize}
Both requests are handled using a deterministic algorithm. The time complexity of \textsc{Add-Edge} is $O(n^2 d^2)$, and of \textsc{Remove-Edge} is $O(nd)$.
\end{lem}

It is worth noting that the proposed implementation of the data structure does not depend on $\gamma$ and $\beta$, and it is up to the user of the data structure to respect the constraints under which \textsc{add-edge} can be invoked.

\begin{proof}[ of Lemma \ref{lemma:oracle}]
The implementation of \textsc{Add-Edge} is given in Algorithm \ref{alg:add-edge}, and the implementation of \textsc{Remove-Edge} in Algorithm \ref{alg:remove-edge}. In Algorithm \ref{alg:add-edge}, we use the following definition of an \emph{alternating walk}.

\begin{defn}[Alternating walk]
    \label{def:alternating}
    Given a subsets of edges $E_1, E_2 \subseteq E(D)$, we say that a sequence of vertices $(v_1, v_2, \ldots, v_k)$, which are not necessarily distinct, forms an \emph{$(E_1, E_2)$-alternating walk} $W$ from $v_1$ to $v_k$ if $(v_{2i+1}, v_{2i+2}) \in E_1$ for every $0 \le i \leq  k/2 - 1$, and $(v_{2i+1}, v_{2i}) \in E_2$ for every $1 \le i < k/2$.

    \noindent
    \textbf{Note:} the direction of the edges along $W$ does not form a directed walk. In particular, they form what is known as an \emph{anti-directed} walk.
\end{defn}

Other than the subset of edges $H$, the data structure  maintains a subset $B \subseteq D \setminus H$ of \emph{buffer} edges and two sets of vertices -- the \emph{saturated} ones denoted by $\Sat$ which are sinks of many edges in $H\cup B$, and those with a low number of out-neighbours which are not saturated, denoted by $\Low$. These sets are initially empty. In the description of the algorithm we use $F$ to denote the set of edges $H \cup B$, always with respect to the current $H$ and $B$.

\begin{algorithm}[h]
    \Input{
        $v \in V(D)$ with $\dout_H(v) < \lfloor d/2 \rfloor$

        \vspace{0.5em}
    }

    \If{$v \in \Low$}{
        Choose any out-edge $e = (v,w) \in B$ of $v$

        Add $e$ to $H$

        Remove $e$ from $B$

        \Return{$e$}
    } \Else {        
        Choose any out-edge $e = (v,w) \in D \setminus F$ of $v$ such that $w \notin \Sat$

        Add $e$ to $H$

        \lIf{$\din_F(w) \ge d/10$}{add $w$ to $\Sat$}
        
        \While{there exists $x \in V(D) \setminus \Low$ such that $\dout_D(x, \Sat) \ge d / 4$}{
            Add $x$ to $\Low$
            
            \While{$\dout_F(x) < \lfloor d/2 \rfloor$}{
                Find a $(D \setminus F, B)$-alternating walk $W$ (see Definition \ref{def:alternating}) from $x$ to some $y$ with $\din_{F}(y) < \lfloor d/5 \rfloor$
                
                \For{$e \in P$}{
                    \lIf{$e \in D \setminus F$}{add $e$ to $B$}
                    \lElse{remove $e$ from $B$}
                }
    
                \lIf{$\din_{F}(y) \ge d/10$}{add $y$ to $\Sat$}
            }            
        }

        \Return{e}
    }
  \caption{\textsc{add-edge$(v)$}}
  \label{alg:add-edge}
\end{algorithm}

\begin{algorithm}[h]
    \Input{
        edge $e = (v,w) \in H$

        \vspace{0.5em}
    }
    
        Remove $e$ from $H$

        \lIf{$v \in \Low$}{add $e$ to $B$}
        
        \ElseIf{$w \in \Sat$ and $\din_H(w) < d/10$}{            
            Remove $w$ from $\Sat$
            
            \While{there exists $x \in \Low$ with $\dout_D(x, \Sat) < d/4$}{
                Remove from $B$ all the out-edges of $x$

                Remove $x$ from $\Low$

                Remove every $y \in \Sat$ with $\din_F(y) < d/10$ from $\Sat$
            }
        }
  \caption{\textsc{remove-edge$(e)$}}
  \label{alg:remove-edge}
\end{algorithm}

\paragraph{Invariants.} Let us, for now, assume that lines 2, 7, and 13 in $\textsc{add-edge}$ are well-defined, that is, a desired edge or a walk exist whenever these lines are executed. More formally, we could say that whenever we reach one of these lines and a desired edge/walk does not exist, we abandon the current request and reset the internal state to how it was right before the request. Then, before handling each request (either \textsc{add-edge} or \textsc{remove-edge}), the following properties hold:
\begin{enumerate}[(P1)]
    \item \label{p:sat} $\Sat = \{v \in V(D) \colon \din_F(v) \ge d/10 \}$,
    \item \label{p:low} $\Low = \{v \in V(D) \colon \dout_D(v, \Sat) \ge d/4 \}$, and
    \item \label{p:F_low} if $v \in \Low$ then $\dout_F(v) = \lfloor d/2 \rfloor$.
\end{enumerate}
Moreover, the following holds throughout each execution of \textsc{add-edge}:
\begin{enumerate}[(C1)]    
    \item \label{p:satc} if $v \in \Sat$ then $\din_F(v) \ge d/10$,
    \item \label{p:lowc} if $v \in \Low$ then $\dout_D(x, \Sat) \ge d/4$,
    \item \label{p:buffer} if  $v \notin \Low$ then $\dout_B(v) = 0$, and
    \item \label{p:F} $\dout_F(v) \le \lfloor d/2 \rfloor$ and $\din_F(v) \le \lfloor d/5 \rfloor$ for every $v \in V(D)$.
\end{enumerate}
These properties follow from the description of the algorithm and the fact that they are trivially satisfied before the first request. We now show a useful corollary of \ref{p:satc}--\ref{p:F}.

\begin{claim} \label{claim:low}
    Throughout the execution of \textsc{add-edge} we have $|\Low| < \beta n / 12$.
\end{claim}
\begin{proof}
    Suppose, towards a contradiction, that there exists a series of requests which results in $|\Low| \ge \lfloor \beta n / 12\rfloor$ at some point. Consider the first moment when this happens, in which case we have that equality holds, and note that this has to be during an execution of \textsc{add-edge}, as during the execution of \textsc{remove-edge} no vertices are added to $\Low$. Then
    \begin{align*}
        |\Sat| d / 10 \stackrel{\ref{p:satc}} \le \din_F(\Sat) &\le e(F) = e(H) + e(B) < \beta d n / 120 + e(B) \\ &\stackrel{\ref{p:buffer}}{\leq} \beta d n / 120 + \dout_B(\Low) \le \beta d n / 120 + |\Low|d \le (1 + 1/10) |\Low| d,
    \end{align*}
    where the last inequality follows from the assumption on the size of $\Low$. Therefore,
    $$
        |\Sat \cup \Low| \le 12 |\Low| \leq \beta n.
    $$
    On the one hand, by the assumption on $D$ we have
    $$
        \dout_D(\Low, \Sat) \le e_D(\Low \cup \Sat) \le |\Low \cup \Sat| \gamma d < 12 \gamma |\Low| d,
    $$
    and on the other hand, by \ref{p:lowc},
    $$
        \dout_D(\Low, \Sat) > d |\Low| / 4,
    $$
    which leads to a contradiction for $\gamma<1/48$.
\end{proof}

\paragraph{Correctness.} We need to check that a returned edge $e = (v,w)$ has the required properties, that is, $\din_H(w) < \lfloor d/5 \rfloor$. If $v \in \Low$ (line 1), then this follows from \ref{p:F}, $F = H \cup B$, and $e \in B$. Otherwise, it follows from $w \notin \Sat$ (line 7) and \ref{p:sat}. 

\paragraph{The algorithm {\sc add-edge} is well-defined.} As lines 2 and 7 happen before any changes to the internal state of the data structure, we can use \ref{p:sat}--\ref{p:F_low} to prove these steps are well-defined. The existence of an edge $e = (v,w) \in B$ in line 2 follows from $\dout_H(v) < \lfloor d/2 \rfloor$ (the assumption of $\texttt{add-edge}$), \ref{p:F_low}, and $F = H \cup B$. The existence of a required edge in line 7 follows from \ref{p:F}, which implies $\dout_{D \setminus F}(v) \ge d - \lfloor d/2 \rfloor \ge d/2$ (recall $D$ is $d$-out-regular), and \ref{p:low} together with $v \notin \Low$. Showing that line 13 is well-defined is more complicated, and this is what we do next.

\begin{claim}
    Line 13 is well-defined.
\end{claim}
\begin{proof}
Throughout the argument we use the subscript $0$, as in $\Sat_0$ for example, to denote the state of a set as it was just before the current invocation of \textsc{add-edge}. Consider some point during the execution of the algorithm when it has reached line 13. We need to show that, at that point, a desired alternating walk exists. Suppose, towards a contradiction, that this is not the case.

Let $\Delta = B \setminus B_0$ and note that $F\setminus F_0 = \Delta \cup \{e\}$. Let $Y \subseteq V(D) \setminus \Sat_0$ be the set of all $y \in V(D) \setminus \Sat_0$ such that there exists an $(D \setminus F, B)$-alternating walk from $x$ to $y$. As $x$ was added to $\Low$ in line 11, we know $x \notin \Low_0$, and hence by \ref{p:low} we have $\dout(v,\Sat_0)<d/4$. Together with $\dout_F(x) < d/2$ (condition in line 12), this implies $Y$ is non-empty because $\dout_{D\setminus F}(x,V(D)\setminus \Sat_0)\geq d-d/2-d/4$. Set $X = \mathrm{In}_\Delta(Y)$ (as $x \in X$, it is non-empty as well). By \ref{p:buffer} and the definition of $\Delta$, since $\Delta$ does not contain edges emanating from $\Low_0$, we have $X \subseteq \Low \setminus \Low_0$.

Observe that $\din_F(y) \ge \lfloor d/5 \rfloor$ for every $y \in Y$, as otherwise a desired path exists. From 
$$
    \din_F(y) = \din_{F_0}(y) + \din_{\Delta}(y) + \mathbb{1}_{y=w}
$$
and $y \notin \Sat_0$, we conclude
$$
    \din_{\Delta}(y) \ge \lfloor d/5 \rfloor - 1 - d/10 > d/11.
$$
This implies
$$
    |Y| d / 11 \le \din_{\Delta}(Y, X) \le d|X|,
$$
thus $|Y| \le 11 |X|$.

Next, note that for every $y \in \textrm{Out}_{D \setminus F}(X)$ there exists an $(D \setminus F, B)$-alternating walk from $x$ to $y$, thus $\textrm{Out}_{D \setminus F}(X) \subseteq \Sat_0 \cup Y$ as otherwise we get a contradiction with the choice of $Y$. On the one hand, from \ref{p:F}, $X \cap \Low_0 = \emptyset$, and \ref{p:low},  we have
$$
    e_D(X\cup Y) \ge \dout_{D \setminus F}(X, Y) \ge \sum_{x \in X} \dout_{D \setminus F}(x) - \dout_{D}(x, \Sat_0) \ge |X|d/4.
$$
On the other hand, since $|Y|\leq 11|X|$ and $X \subseteq \Low$, by Claim \ref{claim:low} we conclude $|X \cup Y| < \beta n$, thus by the assumption on $D$ we have
$$
    e_D(X\cup Y) < 12\gamma|X| d,
$$
which is a contradiction for $\gamma < 1/48$.
\end{proof}

\paragraph{Complexity.} Finding an alternating walk in line 14 can be done using \textsc{BFS}, which takes $O(nd)$ time. Within one call of \textsc{add-edge} this line is executed at most $O(nd)$ times, which gives $O(n^2 d^2)$ and this dominates complexity coming from any other step of \textsc{add-edge}. With careful bookkeeping, \textsc{remove-edge} can be implemented in $O(nd)$.
\end{proof}
\section{Disjoint paths via Edge-Oracle}

\label{sec:main_result}

In this section we prove Theorem~\ref{thm:main}.

\paragraph{Pre-processing $G$.} Suppose that $G$ is an undirected graph. If $d$ is even, we find an Eulerian trail in $G$ and orient the edges along this trail to obtain a $(d/2)$-regular digraph $D$. Note that if $G$ is a $(\beta, \gamma)$-expander, then so is $D$, by definition. If $d$ is odd, we first remove a perfect matching (the existence follows from Tutte's theorem and the fact that $G$ is $d$-edge-connected) and then repeat the previous procedure. Note that in this case we end up with a digraph which is not quite a $(\beta, \gamma)$-expander, but the upper bounds on $e_D(S)$ (see Definition \ref{def:expander}) hold up to an additive factor of $O(d)$, which is negligible thus we will not concern ourselves with it.



By the previous discussion, we can assume now that $G$ is an $(\beta, \gamma)$-expanding $d$-regular digraph. Next, we split $G$ into three disjoint (spanning) subgraphs $G = G_1 \cup G_2 \cup G_3$, where $G_1$ and $G_2$ are $d'$-regular for $d' = \lfloor d / 10 \rfloor$, and $G_3$ is consequently $(d - 2d')$-regular. This can be done as follows: Create an auxiliary bipartite graph $B = (V_1 \cup V_2, E_B)$, where each $V_i$ is a copy of $V(D)$ and there is an edge between $v \in V_1$ and $w \in V_2$ if $(v,w) \in D$. Then $B$ is $d$-regular, hence its edges can be decomposed into $d$ perfect matchings (in polynomial time). Take any $d'$ perfect matchings and assign the corresponding edges to $G_1$, and another $d'$ (different) perfect matchings and assign the corresponding edges to $G_2$.


\paragraph{Internal state.}
We maintain an initially empty sets of edges $H_1 \subseteq G_1, H_2 \subseteq \overleftarrow{G_2}$, and $H_3 \subseteq G_3$, and make use of two instances of \textsc{Edge-Oracle} data structure, $\texttt{out-oracle} = \textsc{Edge-Oracle}(G_1)$ and $\texttt{in-oracle} = \textsc{Edge-Oracle}(\overleftarrow{G_2})$, noting that both $G_1$ and $\overleftarrow{G_2}$ are $d'$-regular and for every $S \subseteq V(G_i)$ of size $|S| \le \beta n$, $i \in \{1,2\}$, we have
$$
e_{G_i}(S) \le \gamma|S|2d=   20\gamma |S| d'\leq \frac{1}{50}|S|d',
$$
as $\gamma<\frac{1}{1000}$, so the conditions of Lemma~\ref{lemma:oracle} are satisfied for $G_1$ and $\overleftarrow{G_2}$. Recall that we can add and remove edges from $H_i$ using the respective oracle as long as $|H_i|\leq c nd$ where $c= \beta/1200$.

\paragraph{Algorithm.} We process $\textsc{Find-Path}(a,b)$ as follows:
\begin{enumerate}
    \item Let $(H_a, V_a) = \textsc{Oracle-BFS}(a, \texttt{out-oracle})$ and $(H_b, V_b) = \textsc{Oracle-BFS}(b, \texttt{in-oracle})$.
    \item Find a shortest path $P'$ from $V_a$ to $V_b$ in $G_3 \setminus H_3$. Let $a' \in V_a$ be the starting point of $P'$, and $b' \in V_b$ the ending point.
    \item Take a shortest path $P_a$ from $a$ to $a'$ in $H_a$ and a shortest path $P_b$ from $b$ to $b'$ in $H_b$. The desired path $P$ is then the concatenation of $P_a$, $P'$, and $\overleftarrow{P_b}$.
    \item \textbf{(Internal update)} For every $e \in H_a \setminus E(P_a)$ call $\textsc{Remove-Edge}(e)$ on \texttt{out-oracle}. Similarly, for every $e \in H_b \setminus E(P_b)$ call $\textsc{Remove-Edge}(e)$ on \texttt{in-oracle}. Set $H_1 := H_1 \cup E(P_a)$, $H_2 := H_2 \cup E(P_b)$, and $H_3 := H_3 \cup E(P')$.
\end{enumerate}

\begin{algorithm}[h]   
    \Input{
        a vertex $v \in V$ \\
        an instance \texttt{oracle} of \textsc{Edge-Oracle} data structure

        \vspace{0.5em}
    }

    \texttt{q} = \textsc{Queue}(v)

    $V' = \{v\}$
    
    $H' = \emptyset$

    \While{\emph{\texttt{q}} is not empty, $|V'| \le \beta n/5 $, and $e(H') < c n d/2$}{
        $u = $ $\texttt{q}.\textsc{dequeue}()$


        \For{$i = 1, \ldots, \lfloor d'/4 \rfloor$}{
            $(u,w) = \texttt{oracle}.\textsc{add-edge}(u)$
        
            \If{$w \not \in V'$}{
               $\texttt{q}.\textsc{enqueue}(w)$

                Add $w$ to $V'$
            }

            Add the edge $(u,w)$ to $H'$
        }
    }

    \Return{$(H', V')$}
    
  \caption{\textsc{Oracle-BFS}$(v, \texttt{oracle})$}
  \label{alg:oracle-bfs}
\end{algorithm}

\noindent
Processing $\textsc{Remove-Path}(P)$ is significantly simpler:
\begin{enumerate}
    \item For each $e \in H_1 \cap E(P)$ call $\textsc{Remove-Edge}(e)$ on \texttt{out-oracle}. For each $e \in H_2 \cap E(\overleftarrow{P})$ call $\textsc{Remove-Edge}(e)$ on \texttt{in-oracle}.
    \item Set $H_1 := H_1 \setminus E(P)$, $H_2 := H_2 \setminus E(\overleftarrow{P})$, and $H_3 := H_3 \setminus E(P)$.
\end{enumerate}

\noindent
Note that \textsc{Add-Edge} and \textsc{Remove-Edge} are called $O(nd)$ times, which results in $O(n^3 d^3)$ and this dominates all other steps of the algorithm.

In the rest of the proof we show that the algorithm for finding new paths is well-defined. We denote by $\mathcal P_s(v)$ the number of paths in $\mathcal P$ which start at $v$, and by $\mathcal P_e(v)$ the number of those which end at $v$. Recall that, by the rules of the game, we always have $|\mathcal P_s(v)|,|\mathcal{P}_e(v)|\leq d/200$, and that $|\mathcal P|\leq r=\varepsilon nd/\log n,$ for a small enough $\varepsilon=\varepsilon(\gamma,\beta)$.

\paragraph{Invariants.} Before each request the following properties hold:
\begin{enumerate}[(P1)]
    \item \label{mainp:union} $\bigcup_{P \in \mathcal{P}} E(P) = H_1 \cup \overleftarrow{H_2} \cup H_3$,
    \item \label{mainp:size} $|H_1|, |H_2| \le |\mathcal{P}| \log n\leq c nd/2$ and $|H_3| \le 300 |\mathcal{P}| / \beta$,
    \item \label{mainp:out} $\dout_{H_1}(v) \le \din_{H_1}(v) + \mathcal{P}_s(v)$ and $\dout_{H_2}(v) \le \din_{H_2}(v) + \mathcal{P}_e(v)$ for every $v \in V$,
    \item \label{mainp:in} $\din_{H_i}(v) \le \lfloor d' / 5 \rfloor$ for every $v \in V$ and $i \in \{1,2\}$.
\end{enumerate}
The property \ref{mainp:in} is responsible for bounding the number of times each vertex is used as an inner vertex of paths in $\mathcal{P}$, and is the heart of the proof!

The property \ref{mainp:union} follows from the description of the algorithm. To see that the property \ref{mainp:out} holds, note that every out-edge of a vertex in a path which belongs to $H_1$ is preceded by an in-edge of $v$ in $H_1$, unless it is the first vertex which is accounted by $\mathcal{P}_s(v)$. Note that we do not have equality as the out-edge and the in-edge might not belong to the same $H_i$. The second inequality is \ref{mainp:union} is obtained analogously, taking into account that the edges in $H_2$ are reversed compared to $G$. Since $H_i$ for $i \in \{1, 2\}$ is obtained through \textsc{Add-Edge} requests to an oracle, property \ref{mainp:in} holds by the definition.
It remains to show that \ref{mainp:size} holds. This is done in the following two claims.

    

    






    



    



            
            
    

\begin{claim}\label{cl:BFS size}
    Assuming that every call of \textsc{Add-Edge} in line 7 of Algorithm \ref{alg:oracle-bfs} is valid (that is, assumptions of \textsc{Add-Edge} are satisfied), after \textsc{Oracle-BFS} finishes, every vertex in $V'$ is at distance at most $\log n$ from $v$, using only edges in $H'$. Furthermore, $|V'|\geq \beta n/5$.
\end{claim}
\begin{proof}
    Let $T$ be the tree rooted at $v$ consisting of edges $(u,w)$ added to $H'$, such that $w$ was not in $V'$ when the edge $(u,w)$ was returned by the oracle. Let $\ell$ be the depth of $T$, and note that vertices in  \texttt{q} are all contained in levels $\ell-1$ and $\ell$. Importantly, any vertex in $V'$ which is not in \texttt{q} has exactly $\lfloor d'/4 \rfloor > d'/5$ out-edges in $H'$. Denote by $S_i$ the $i$-th level, and let us show that $S_{\ell-1}$ is at distance at most $\log n-1$ from the root $v$, which will give that every vertex in $V'=V(T)$ is at distance at most $\log n$ from $v$. Let $S_{<i}$ be the union of the vertices in the first $i-1$ levels. It suffices to show that $|S_i|\geq 2|S_{<i}|$ for all $i<\ell$.
    
    Suppose that is not the case for some $i$, that is, $|S_i| < 2|S_{<i}|$. From $\textrm{Out}_{H'}(S_{<i}) \subseteq S_{<i} \cup S_i$ we get
    \begin{equation} \label{eq:e_Si}
        e_{G}(S_{<i} \cup S_i) \ge e_{H'}(S_{<i} \cup S_i) \ge \dout_{H'}(S_{<i}) \ge |S_{<i}| d'/5. 
    \end{equation}
    As $|S_{<i} \cup S_i| < 3|S_{<i}|$ and $|V'| \le \beta n$, by the assumption that $G$ is an expander we have
    $$
        e_G(S_{<i} \cup S_i) \le \gamma \cdot 3 |S_{<i}| \cdot 2 d,
    $$
    which contradicts \eqref{eq:e_Si} as $\gamma<1/300$ and $d'=d/10$.
    
    To finish, we have to show that the algorithm terminates when $|V'|\geq \beta n/5$. If that was not the case, then either $H'$ contans at least $cnd/2$ edges, or \texttt{q} is empty. If the former is true, then $H'\subseteq G$ is a digraph whose vertex set $V'$ is of size at most $|V'|\leq \beta n/5$ and which has at least 
    $$e(H')\geq cnd/2=\frac{\beta}{2400}nd > \gamma|V'| \cdot 2d$$
    edges,
    again contradicting the assumption that $G$ is an expander. Finally, it is easy to see that \texttt{q} cannot be empty, as then again $H'$ is a graph on at most $\beta n/5$ vertices with every vertex of degree at least $d'/5$ -- thus a contradiction as in the previous case.
\end{proof}

\begin{claim}\label{cl:G3}
    Let $V_a, V_b \subseteq V$ be sets of size at least $\beta n$. Then for any subset of edges $R \subseteq G_3$ of size $|R| \le \beta n d / 50$, $G_3 \setminus R$ contains a path $P$ from $V_a$ to $V_b$ of length at most $300 / \beta$.
\end{claim}
\begin{proof}
    Let $G' = G_3 \setminus R$. Note that for every $S \subseteq V(G)$ of size $\beta n/5 \le |S| \le n/2$ we have
    $$
        \dout_{G'}(S, V \setminus S) \ge (d - 2d')|S| - e_G(S) - |R| \ge 4 d |S| / 5 - 2 d |S|/3  - |S| d / 10 \geq |S| d / 30,
    $$
    thus
    $$
        |\textrm{Out}_{G'}(S)| \ge |S|/30 \geq  \beta n /150.
    $$
    Therefore, there are more than $n/2$ vertices which can be reached from $V_a$ within $150 / \beta$ steps in $G'$. Note that this estimate can be improved, but for us this simple bound suffices.
    
    The same argument shows that there are more than $n/2$ vertices from which a vertex in $V_b$ can be reached within $150/\beta$ steps. Therefore, there exists a vertex which can be both reached from $V_a$ in $150/\beta$ steps, and which can reach $V_b$ in $150/\beta$ steps, implying a desired path from $V_a$ to $V_b$ exists.
\end{proof}

With Claim~\ref{cl:BFS size} and Claim~\ref{cl:G3} at hand, we clearly see that in Step 3 of the main algorithm, the found paths $P_a,P_b$ are always of length at most $\log n$, while $P'$ is of length at most $300/\beta$ (and is disjoint of $H_3$ as we let $R=H_3$, and inductively we have $|R|=|H_3|\leq 300|\mathcal P|/\beta\leq \beta nd/50$).
Since we only add the edges of those paths to $H_1,H_2$ and $H_3$ respectively, \ref{mainp:size} follows. 

To complete the proof, we show that line 7 of Algorithm~\ref{alg:oracle-bfs} is valid. Indeed, before each execution of Algorithm~\ref{alg:oracle-bfs}, by \ref{mainp:size} we have $|H_1|,|H_2|\leq c nd/2$, hence by the definition of the oracle, we can make at least $c nd/2$ new requests of the form \textsc{add-edge}(v), as long as $\dout_{H_i}(v)\leq d'/2$. By \ref{mainp:out} and \ref{mainp:in}, we further know that before we call Algorithm~\ref{alg:oracle-bfs}, we always have $\dout_{H_1}(v) \le \din_{H_1}(v) + \mathcal{P}_s(v)\leq d'/5+d'/20\leq d'/4$. Therefore, since by the description of the algorithm we request a vertex $v$ at most $\lceil d'/4\rceil$ times, we have that during the execution of the algorithm we always have $\dout_{H_i}(v)\leq d'/2$, which finishes the proof.


\subsection{Sketch of proof for Theorem~\ref{thm:ndL}}
Recall that by the Expander mixing lemma, in every $(n,d,\lambda)$-graph $G$ it holds that every set of vertices of size $\alpha n$ has average degree at most $\alpha d+\lambda$ (see Section 9.2 in \cite{alon2016probabilistic}). Hence $G$ is a $(\beta, \gamma)$-expander for $\beta=1/10^4$ and $\gamma= 1/5000$, for say $\lambda < d/10^4$.

The proof of Theorem~\ref{thm:ndL} is essentially the same as the proof of Theorem~\ref{thm:main} presented in this section, up to Claim~\ref{cl:BFS size}, where we can use the stronger expansion properties of $(n,d,\lambda)$-graphs to show that every vertex in $V'$ is in fact at distance at most $O(\frac{\log n}{\log(d/\lambda)})$ from $v$. Indeed, even after deleting a  constant fraction of edges at each vertex in an $(n,d,\lambda)$-graph, each set up to a certain linear size expands by a factor of $\Theta(\frac{d^2}{\lambda^2})$ (which again follows from the Expander mixing lemma, see Lemma 3.6 in \cite{draganic2022rolling}). Consequently, we have $|S_i|\geq \Theta(\frac{d^2}{\lambda^2})|S_{<i}|$, which evidently gives the required distance between $v$ and every vertex in $V'$.

\newpage

\bibliographystyle{abbrv}

\end{document}